\documentclass[12pt]{article}

\usepackage{graphicx}
\usepackage{geometry}
\usepackage{amssymb}
\usepackage{amsmath ,bm}

\newcommand{\bey}{\begin{eqnarray}}
\newcommand{\eey}{\end{eqnarray}}
\newcommand{\beq}{\begin{equation}}
\newcommand{\eeq}{\end{equation}}

\newtheorem{thm}{Theorem}[section]

\newtheorem{cor}{Corollary}[section]

\begin{document}


\begin{center}

{\large \bf Maximum Likelihood Criterion for Non-nested Model Selection}

\bigskip \bigskip
Min Tsao 
\vspace{0.1in}
\\{\small Department of Mathematics \& Statistics, University of Victoria, Canada}

\end{center}

\bigskip

\begin{abstract}
Penalization is a widely used approach to model selection with roots in information theory and Bayesian inference. We study a model selection problem involving non-nested candidate models for which penalization is counterproductive. We propose a Maximum Likelihood Criterion for this non-nested setting that selects the candidate model with the highest maximum likelihood. This criterion does not take into consideration the number of parameters of a candidate model. It is well-suited for situations where all candidate models are regarded as equal with no preference for models having fewer parameters. We establish the consistency of this criterion and compare its performance with that of existing penalization-based criteria.

\vspace{0.2in}

\vspace{0.1in}
\noindent \textsc{MSC 2020 subject index:}   Primary 62J05, 62J12; secondary 62F03.

\end{abstract}

\section{Introduction}
Penalization is a commonly used strategy in statistics with applications in many areas such as density estimation, functional data analysis, Bayesian statistics and regularized regression. In model selection, this strategy is formulated via penalizing an objective function, usually the log-likelihood of a model, by a term proportional to an $\ell_p$ norm of the parameter vector of the model.  Akaike Information Criterion (AIC) (Akaike, 1974) and Bayesian Information Criterion (BIC) (Schwarz, 1978) are two well-known examples with penalty terms proportional to the $\ell_0$ norm. 
These penalization-based methods are effective in alleviating the negative impact of overfitting in model selection. Nevertheless, we argue that overfitting is of concern only in nested model selection problems, such as regression model selection, where better fit comes at the expense of increased model complexity. There are situations where practitioners want to identify the best model for observed data among a small set of non-nested models. In such non-nested situations, the notion of relative model complexity and the issue of overfitting become irrelevant. For example, when an engineer needs to decide which model among Exponential, Log-normal, and Weibull models that best describes a set of observed reliability data, all three models are equally important. He is not concerned about overfitting and does not have a preference for the Exponential model just because it has fewer parameters. In such situations, model selection should be based on the fit of candidate models alone.

In this note, we consider the setting of a small set of parametric candidate models that are non-nested and equally important, and no preference or advantage is assigned to models having fewer parameters. We propose a Maximum Likelihood Criterion (MLC) that selects the model with the best fit; that is, MLC selects the model with the highest maximum likelihood. It may seem unconventional to compare maximized likelihood values of different parametric models, as such comparisons fall outside the traditional scope of maximum likelihood theory and practice. Furthermore, AIC and BIC do not compare the maximum likelihood of different models directly. Instead, they compare their penalized maximum log-likelihoods which have interpretations related to model distance and plausibility measures from information theory and Bayesian inference. It is such interpretations that justify comparing the penalized maximum log-likelihood of different models. We assume that the true model that generated the data is in the set of candidate models. Under this assumption, we show that MLC is consistent in that the probability that it selects the true model approaches one as the sample size goes to infinity. In other words, asymptotically the true model has the highest maximum likelihood. This provides an asymptotic justification for comparing the maximum likelihood of non-nested models directly for the purpose of model selection. It also gives a justification for ignoring the number of parameters in such candidate models. Our proof of consistency is based on the positivity of Kullback-Leibler divergence (Kullback  and  Leibler, 1951). We compare MLC with AIC and BIC in terms of selection accuracy via simulation, and demonstrate that in the setting where candidate models are non-nested and considered equal regardless their number of parameters, the MLC is a fairer criterion for model selection than AIC and BIC. We also give a real-data example to compare AIC, BIC and MLC.

\section{Maximum Likelihood Criterion}

Consider a set of $M$ candidate models, denoted by their families of density functions $\mathcal{M}_1 = \{f_1(x; \theta_1): \theta_1 \in \Theta_1\}, \dots, \mathcal{M}_M = \{f_M(x; \theta_M): \theta_M \in \Theta_M\}$. We assume these models are non-nested in the sense that for $j\neq k$ and any $\theta_j \in \Theta_j$ and $\theta_k \in \Theta_k$, $f_j(x, \theta_j)$ and $f_k(x, \theta_k)$ are not equal almost everywhere. Denote by $\mathbf{X}_n = \{X_1, \dots, X_n\}$ a random sample of size $n$ from a data generating distribution $f_T(x, \theta^*_T) \in \mathcal{M}_T \in \{\mathcal{M}_1, \dots, \mathcal{M}_M\}$.  In the context of model selection with observed sample $\mathbf{X}_n$, the data generating distribution $f_T(x, \theta^*_T)$ and the family it belongs to $\mathcal{M}_T$ are both unknown. We will refer to $\mathcal{M}_T$ as the true model and the other ($M-1$) models as misspecified models. Our objective is to identify the true model among the $M$ candidate models using $\mathbf{X}_n$.

The proposed Maximum Likelihood Criterion (MLC) selects the candidate model with the highest maximum likelihood (or equivalently the highest maximum log-likelihood) as the true model. Specifically, denote the MLC selected model by $\text{MLC}(\mathbf{X}_n)$, we have
\begin{equation}
\text{MLC}(\mathbf{X}_n) =  \arg\max_{\{M_j: 1\leq j \leq M \}} \ell_j(\hat{\theta}_j), \label{def}
\end{equation}
where $  \ell_j(\hat{\theta}_j)=\sum_{i=1}^n \log f_j(X_i; \hat{\theta}_j)$ and $\hat{\theta}_j$ is the maximum likelihood estimator (MLE),
  \[
        \hat{\theta}_j = \arg\max_{\theta \in \Theta_j} \sum_{i=1}^n \log f_j(X_i; \theta).
    \]

The MLC is consistent in that
\begin{equation}
    \lim_{n \to \infty} P\left( \text{MLC}(\mathbf{X}_n) = \mathcal{M}_T \right) = 1.  \label{consistency}
\end{equation}
We now prove the consistency (\ref{consistency}) for $M=2$. The generalization to any fixed $M>2$ is straightforward. We first prove the following theorem which establishes the dominance of the likelihood of the data generating distribution over that of misspecified models.

\begin{thm} \label{thm1}
    Let $f(x)$ be the true probability density function and $\{g(x;\theta): \theta \in \Theta\}$ be a family of density functions. Assume the following conditions hold: 
    \begin{enumerate}
        \item[(C1)] The true density $f(x)$ is not in the $g(x;\theta)$ family in the sense that for any $\theta\in \Theta$,  $f(x)$ and $g(x;\theta)$ are not equal almost everywhere.
        \item[(C2)] $\Theta$ is compact.
        \item[(C3)] Function $\log g(x;\theta)$ is continuous in $\theta$.
        \item[(C4)] There exists an integrable function $h(x)$ such that $|\log g(x;\theta)| \le h(x)$ for all $\theta \in \Theta$.
    \end{enumerate}
    Let $\{X_1, \dots, X_n\}$ be a random sample from $f(x)$. Denote by $\hat{\theta}_n$ the MLE of $\theta$ for model $g(x,\theta)$, i.e., $ \hat{\theta}_n = \arg\max_{\theta \in \Theta} \sum_{i=1}^n \log g(X_i; \theta).$ Then,
    \begin{equation}
        \lim_{n \to \infty} \mathbb{P}\left( \sum_{i=1}^n \log f(X_i) \ge \sum_{i=1}^n \log g(X_i; \hat{\theta}_n) \right) = 1. \label{r1}
    \end{equation}
\end{thm}

\vspace{0.2in}

\noindent {\bf Proof of Theorem \ref{thm1}.} 
    The Kullback--Leibler (KL) divergence between $f$ and $g(\cdot;\theta)$ is
    \[
        D_{\text{KL}}(f \| g(\cdot;\theta)) = \mathbb{E}_f\left[ \log \frac{f(X)}{g(X;\theta)} \right] = \mathbb{E}_f[\log f(X)] - \mathbb{E}_f[\log g(X;\theta)].
    \]
    Under conditions (C1)-(C2),  the minimum KL divergence over $\Theta$ is positive. That is, let
        \[
        \theta^* = \arg\max_{\theta \in \Theta} \mathbb{E}_f[\log g(X;\theta)].
    \]
    Then,
     \begin{equation}
        \Delta := \mathbb{E}_f[\log f(X)] - \mathbb{E}_f[\log g(X;\theta^*)] = D_{\text{KL}}(f \| g(\cdot;\theta^*)) > 0.  \label{KL}
     \end{equation}
     
\vspace{0.1in}
     Under conditions (C2)-(C4), by the Uniform Law of Large Numbers (Jennrich, 1969), we have
        \begin{equation}
        \sup_{\theta \in \Theta} \left| \frac{1}{n} \sum_{i=1}^n \log g(X_i;\theta) - \mathbb{E}_f[\log g(X;\theta)] \right| \xrightarrow{p} 0.  \nonumber 
    \end{equation}
    Consequently,
    \begin{equation}
        \frac{1}{n} \sum_{i=1}^n \log g(X_i; \hat{\theta}_n) \xrightarrow{p} \mathbb{E}_f[\log g(X;\theta^*)]. \label{2.1}
    \end{equation}
    Also, by the Strong Law of Large Numbers,
    \begin{equation}
        \frac{1}{n} \sum_{i=1}^n \log f(X_i) \xrightarrow{a.s.} \mathbb{E}_f[\log f(X)]. \label{2.2}
    \end{equation}

   Equations (\ref{KL}),  (\ref{2.1}) and  (\ref{2.2})  imply that 
    \begin{eqnarray}
        \frac{1}{n} \sum_{i=1}^n \log f(X_i) - \frac{1}{n} \sum_{i=1}^n \log g(X_i; \hat{\theta}_n) &\xrightarrow{p}& 
        \mathbb{E}_f[\log f(X)] - \mathbb{E}_f[\log g(X;\theta^*)]   \nonumber \\
        &=& \Delta > 0.  \label{2.3}
    \end{eqnarray}
This implies that

\begin{equation}
        \mathbb{P}\left( \sum_{i=1}^n \log f(X_i) \ge \sum_{i=1}^n \log g(X_i; \hat{\theta}_n) \right) \rightarrow 1 \hspace{0.2in} 
        \text{as $n\rightarrow \infty$}. \label{2.7}
    \end{equation}
    which proves (\ref{r1}).    \hfill $\Box$

\vspace{0.2in}
The following corollary gives the consistency of MLC for $M=2$.

\begin{cor} \label{cor1}
    Let  $\{f(x;\theta): \theta \in \Theta_f\}$ and $\{g(x;\theta): \theta \in \Theta_g\}$ be two non-nested family of densities.  Let $\{X_1, \dots, X_n\}$ be a random sample from $f(x, \theta_T)\in \{f(x;\theta): \theta \in \Theta_f\}$. Denote by $\hat{\theta}_{f,n}$ and $\hat{\theta}_{g,n}$ the MLE's of $\theta$ based on the sample where
    \[
        \hat{\theta}_{f,n} = \arg\max_{\theta \in \Theta_f} \sum_{i=1}^n \log f(X_i; \theta) \hspace{0.1in} \text{and}  \hspace{0.1in} 
        \hat{\theta}_{g,n} = \arg\max_{\theta \in \Theta_g} \sum_{i=1}^n \log g(X_i; \theta).
    \]
    Assume conditions (C1)-(C4) hold for $f(x,\theta_T)$ and $\{g(x;\theta): \theta \in \Theta_g\}$. Then,
    \begin{equation}
        \lim_{n \to \infty} \mathbb{P}\left( \sum_{i=1}^n \log f(X_i, \hat{\theta}_{f,n}) \ge \sum_{i=1}^n \log g(X_i; \hat{\theta}_{g,n}) \right) = 1.  \label{2.10}
    \end{equation}
\end{cor}

\noindent {\bf Proof of Corollary \ref{cor1}.} 
By the definition of $\hat{\theta}_{f,n}$, we have
\[
    \sum_{i=1}^n \log f(X_i, \hat{\theta}_{f,n}) \ge \sum_{i=1}^n \log f(X_i; {\theta}_T).
\]  
 This and (\ref{2.7}) imply (\ref{2.10}), and thus the consistency of MLC (\ref{consistency}).
 \hfill $\Box$

\vspace{0.1in}
\textsc{Remark 1.} The above results also lead to a class of consistent penalization-based criteria that includes AIC and BIC as special cases.  To see this, let $\ell_f=\sum_{i=1}^n \log f(X_i)$ and $\ell_g(\hat{\theta}_{n})=\sum_{i=1}^n \log g(X_i; \hat{\theta}_n)$.  By (\ref{2.3}), $n^{-1}\{\ell_f - \ell_g(\hat{\theta}_n) \}$ converges in probability to the minimum KL divergence $\Delta$ which is a positive constant. Thus, $\ell_f - \ell_g(\hat{\theta}_n)\sim n\Delta$ in probability. 
Similarly, for the two families of models in Corollary \ref{cor1},  let $\ell(\hat{\theta}_{f,n})= \sum_{i=1}^n \log f(X_i, \hat{\theta}_{f,n})$ and $\ell(\hat{\theta}_{g,n})= \sum_{i=1}^n \log g(X_i, \hat{\theta}_{g,n})$. 
Then, $n^{-1}\{ \ell(\hat{\theta}_{f,n}) - \ell(\hat{\theta}_{g,n}) \}$ also converges to $\Delta$ and $\ell(\hat{\theta}_{f,n}) - \ell(\hat{\theta}_{g,n})\sim n\Delta$ in probability. 
Define a class of penalized likelihood criteria that selects the model with the minimum penalized log-likelihood 
\begin{equation}
\text{PLC}(\mathbf{X}_n) = \arg\min_{\{f, g\}}  -\ell(\hat{\theta}_{n})+p(d, n),  \label{p.cri}
\end{equation}
where $p(d, n)$ is a positive penalty function and $d$ is the dimension of the parameter vector. This class is consistent under the condition that $p(d, n)=o(n)$. This is because asymptotically the difference between the log-likelihood of the true model $\ell(\hat{\theta}_{f,n})$ and that of the misspecified model $\ell(\hat{\theta}_{g,n})$ is $n\Delta$ which dominates the penalty term of size $o(n)$. Thus, the true model will have a smaller penalized log-likelihood than that of the misspecified model with probability tending to one. This class of penalization-based criteria includes the AIC whose $p(d, n)=2d$ and the BIC for which $p(d, n)=d \log n$. Nevertheless, in the next section, we demonstrate via numerical examples that penalization is unnecessary in this non-nested setting as it introduces bias against models with more parameters. Without a penalty term, the MLC is  fairer in finite sample applications.

\vspace{0.1in}

\textsc{Remark 2.}  In a paper on sparse maximum likelihood estimation of regression models, Tsao (2025) showed that a model of maximum likelihood with $p^*$ or more variables (here, $p^*$ is the number of active variables) is asymptotically correct in that it contains all active variables. The proof of this result in Tsao (2025) is based on a combination of maximum likelihood theory and geometric arguments. The fact that $\ell(\hat{\theta}_{f,n}) - \ell(\hat{\theta}_{g,n})\sim n\Delta$  may be used for an alternative proof based on the positivity of KL divergence. 

\section{Numerical Examples}
In this section, we compare MLC, AIC and BIC through numerical examples. For simplicity of presentation, we consider only three pairs of models in 
Table \ref{tb-1}. 

The first pair is Exponential model and Log-normal model. The second pair is Uniform model and Normal model. These two pairs satisfy conditions in Theorem \ref{thm1} and Corollary \ref{cor1}, so AIC, BIC, and MLC are all consistent for these examples. Note that strictly speaking the parameter spaces for models in these two pairs of examples are not compact but we can assume that they are sufficiently large compact sets so the compactness condition holds. In practical applications, this compactness assumption is supported by the nature of the observed data. For example, when the observed data are lifetimes of a certain brand of lightbulbs and \textsc{Exp$(\lambda)$} model is fitted to the data, the parameter space of \textsc{Exp$(\lambda)$}  is naturally a finite interval, say $\lambda\in [a, b] \subset (0, \infty)$, as the parameter $\lambda$ of the exponential distribution is its mean. When used as a model for the lifetime of lightbulbs, the mean is in a finite interval $[a, b]$. The compactness assumption remains valid even though interval $[a, b]$ is unknown. The third pair  is Chi-square and Gamma models. The non-nested condition is violated for this pair as the former is nested within the latter. We include this pair to show the consequence of this violation.

Column one of Table \ref{tb-1} contains the true model as well as the parameter value(s) used to generate random observations for model selection simulation. Column two contains the alternative misspecified model. Column three shows the sample size used. The last three columns contain the selection error rates of MLC, AIC, and BIC based on 10,000 simulation runs. For each run, we generate a random sample of size $n$ using the true model at the indicated parameter value(s), and then use the three criteria to select a model based on the random sample.  Each error rate is the percentage of times the criterion in the column heading selects the misspecified model out of 10,000 simulation runs. We make the following observations based on the results in Table \ref{tb-1}.

\begin{table} 
\caption{\label{tb-1} Comparison of selection error rates of MLC, AIC and BIC. The first column contains the true models and their parameter values used to generate the observed data. The second column contains the misspecified models. Each entry in the last three columns is the percentage of times the misspecified model is selected by a criterion based on 10,000 simulated random samples from the true model in the same row.} 
\centering
{\small
\begin{tabular}{cccccc} \\ 
\textsc{True model} & \textsc{Misspecified model} & $n$ &\textsc{MLC} & \textsc{AIC} & \textsc{BIC} \\ \hline \hline 
\textsc{Exp$(\lambda)$}   & \textsc{log-normal$(\mu, \sigma^2)$} & 10 & 51.29  & 21.88 & 18.59 \\
 For simulation $\lambda=1$                                            &              & 50 & 14.06 & 7.56 & 3.62\\
                                             &                                                                & 100 & 3.70 & 2.18 & 1.00\\   \hline 
                                             
\textsc{log-normal$(\mu, \sigma^2)$} & \textsc{Exp$(\lambda)$}   & 10 & 13.45  & 48.81 & 55.01 \\
        For simulation $(\mu, \sigma^2)=(0,1)$                                      &                     & 50 & 3.95 & 10.08 & 19.52 \\
                                             &                                                        & 100 & 1.04 & 2.26 & 5.34 \\   \hline \hline 
                                             
\textsc{Unif$[-a, a]$}   & \textsc{N$(\mu, \sigma^2)$} & 10 & 9.39  & 3.87 & 3.36 \\
For simulation $a=1$                                           &                    & 50 & 0.45 & 0.25 & 0.13 \\
                                             &                                                         & 100 & 0.06 & 0.05 & 0.02\\   \hline 
                                             
\textsc{N$(\mu, \sigma^2)$} & \textsc{Unif$[-a, a]$}   & 10 & 58.00  & 77.47 & 80.06 \\
         For simulation $(\mu, \sigma^2)=(0,1)$                                    &  & 50 & 5.15 & 7.47& 10.55 \\
                                             &                                & 100 & 0.22 & 0.34 & 0.52 \\   \hline \hline 

\textsc{$\chi^2_k$} & \textsc{Gamma($\theta, \kappa$)}   & 10 & 100.00  & 50.17 & 47.48  \\
For simulation $k=3 $            &                                          & 50 & 100.00 & 45.68 & 30.66 \\
                                             &                                         & 100 & 100.00 & 46.26 & 26.00\\   \hline 

\textsc{Gamma($\theta, \kappa$)}   & \textsc{$\chi^2_k$} & 10 & 0.00  & 10.45 & 11.40 \\
For simulation $(\theta, \kappa)=(3,1)$            &               & 50 & 0.00  & 0.10 & 0.35 \\
                                             &                                           & 100 & 0.00 & 0.00 & 0.00 \\   \hline \hline 

\end{tabular}
}
\end{table}

\begin{itemize}

\item[1.] \textsc{Exp($\lambda$) vs Log-normal($\mu, \sigma$)}: The true model \textsc{Exp($\lambda$)} and the misspecified alternative \textsc{Log-normal($\mu, \sigma$)} are both models for lifetime data supported on $(0, +\infty)$. When the sample size is $n=10$, all three criteria performed poorly. In particular, MLC has an error rate exceeding 50\%. This is a case where the true model \textsc{Exp($\lambda$)}  has fewer parameters, so the penalization-based criteria AIC and BIC which favour smaller models have smaller but still large error rates. As the sample size $n$ goes to $50$ and  $100$, error rates of all three criteria drop substantially, reflecting the consistency of all three criteria. With larger and larger sample sizes, the likelihood of the true model will dominate that of the misspecified model, as shown in Corollary \ref{cor1}, even if the true model has fewer parameters.

\textsc{Log-normal($\mu, \sigma$) vs Exp($\lambda$)}: The true model is now \textsc{Log-normal($\mu, \sigma$)} which has more parameters. At $n=10$, all three criteria again performed poorly. The penalty terms in AIC and BIC are counterproductive which led to high error rates. As the sample size increases, the error rates drop quickly towards zero.

\item[2.] \textsc{Unif$[-a, a]$ vs N$(\mu, \sigma^2)$}: This is an example where the two candidate models have different supports. Again, the selection error rates of all three criteria are high at the small sample size of $n=10$ but improves quickly as $n$ goes to $50$ and $100$.

\textsc{N$(\mu, \sigma^2)$ vs Unif$[-a, a]$}: One case worth mentioning is MLC at $n=10$. The error rate of MLC in this case exceeds 50\% indicating that the likelihood of the true model with 2 parameters is smaller than that of the misspecified model with only 1 parameter more than 50\% of the time. This illustrates the fact that more parameters do not always lead to higher likelihood even when the model with more parameters is the true model. It shows that MLC does not favour models with more parameters. Indeed, MLC is solely a likelihood-based criterion that does not take into consideration the number of parameters. On the other hand, AIC and BIC does in a way that discriminates against models with more parameters.

\item[3.] \textsc{$\chi^2_k$ vs Gamma($\theta, \kappa$)}: This is an example where the non-nested condition fails, so the three criteria are no longer consistent. Since $\chi^2_k$ is nested within Gamma($\theta, \kappa$), the latter model will always have a higher likelihood even though the former is the true model that generated the data. This explains the 100\% error rate of MLC. The error rates of AIC and BIC are lower than that of MLC but still high at all sample sizes. Since the data is generated by $\chi^2_k$, the better fit of Gamma($\theta, \kappa$) is a result of overfitting. This is an example that shows penalization is helpful in alleviating overfitting in nested situations where more parameters means a better fit.

\textsc{Gamma($\theta, \kappa$) vs $\chi^2_k$}: Since $\chi^2_k$ is nested within Gamma($\theta, \kappa$), the MLC has zero errors when the latter is the true model. Due to their penalty terms against models with more parameters, AIC and BIC have small but not non-zero error rates.

\end{itemize}

To summarize, penalization introduces bias against models with more parameters which is unfair when candidate models are non-nested. It should be acknowledged that penalization brings better accuracy when the true model happens to be the one with fewer parameters, but this does not justify such a bias in the setting where all candidate models are equal with no preference for models having fewer parameters; see the concluding remarks for further discussions. By letting the likelihood alone determine the selection, the MLC provides a fairer alternative to penalization-based criteria.

\section{Precip data analysis}

The precip dataset is a built-in dataset in R that records the average annual precipitation, measured in inches, for 70 cities in the United States and Puerto Rico. The values are climatological averages computed over the 30-year period from 1941 to 1970, and therefore represent long-run rainfall patterns rather than observations from a single year. The dataset exhibits substantial geographic variability, ranging from arid cities in the southwestern United States to humid coastal and southeastern locations. Because the data consist of a single quantitative variable without any explanatory variables, we consider several parametrical candidate models for this dataset.

Figure 1 shows the density histogram of this dataset with fitted Normal and Gamma density functions overlayed on top. The histogram shows that the dataset is mildly skewed. We consider the following five candidate models for this dataset: Normal, Logistic, Laplace, Gamma, and Chi-square models. The first four models are not nested.  Although the Chi-square model is nested within the Gamma model, no other one-parameter models fit better than the Chi-square model, so we include it in the comparison. 
		\begin{figure}  \label{fig-1}
\caption{Histogram of precip dataset with fitted Normal and Gamma density functions overlayed on top of the histogram. The dataset consists of average annual precipitation, measured in inches, for 70 cities in the United States and Puerto Rico.}
		\begin{center}
		\includegraphics[width=10cm]{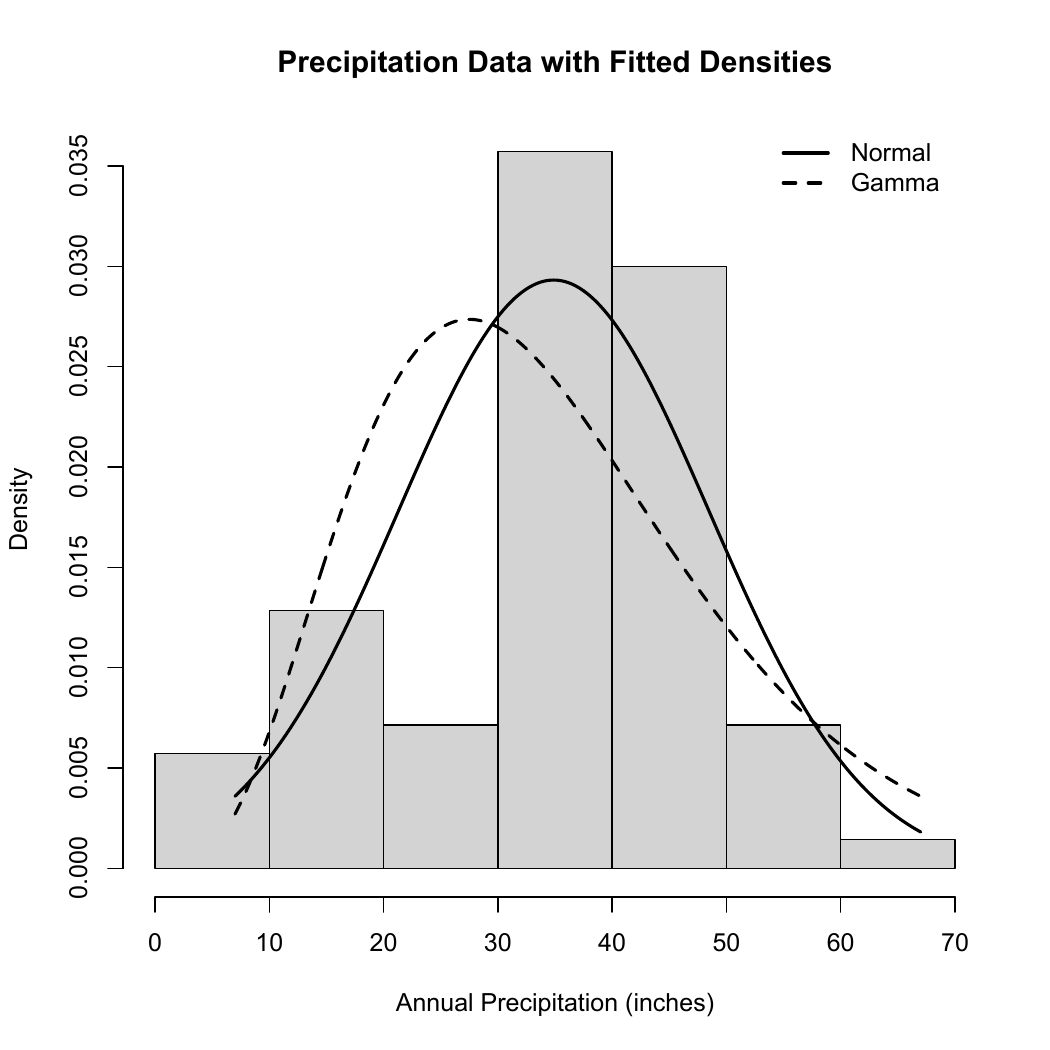}
		\end{center}
		\end{figure}

\begin{table} 
\caption{\label{tb-2} Comparison of MLC, AIC and BIC in model selection for the precip data. Five candidate models are considered: Normal, Logistic, Laplace, Gamma, and Chi-square distributions. All three criteria rank the candidate models in the same order. } 
\centering
{\small
\begin{tabular}{ccccc} \\ 
\textsc{Model} & \textsc{Parameters} &\textsc{MLC} & \textsc{AIC} & \textsc{BIC} \\ \hline \hline 
\textsc{Normal}   & 2        & -282.074  & 568.148 & 572.645 \\
\textsc{Logistic}  & 2        & -282.794  & 569.589 & 574.086 \\
\textsc{Laplace}  & 2       & -283.059  & 570.119 & 574.616 \\
\textsc{Gamma}  & 2       & -288.465  & 580.929 & 585.426 \\
\textsc{Chi-square}   & 1 & -335.769 & 673.538 & 675.786 \\

\end{tabular}
}
\end{table}

Table \ref{tb-2} shows AIC, BIC, and MLC (log-likelihood) values. By MLC, models with higher log-likelihood are better models. By AIC and BIC, models with lower AIC and BIC values are better. For this example, all three criteria rank the five models in the same order with the Normal model at the top and the Chi-square model at the bottom. This is not a mere coincidence; since the whole class of penalized criteria (\ref{p.cri}) are consistent, with sufficiently large sample sizes they will all choose the same model.

\section{Concluding remarks}

The penalty terms of AIC and BIC are most effective in regression model selection where model complexity is effectively measured by the number of predictor variables in the model and there is a clear danger of overfitting in that better and better fit can be reached with more and more variables. In such cases, the penalty terms provide a needed balance between complexity and overfitting. But in cases with only non-nested candidate models, there is no danger of overfitting. Complexity of models is also not well-defined as it can no longer be measured by the number of parameters; parameters in different models have different interpretations, so they and their numbers cannot be compared. We have also seen that more parameters does not necessarily bring better fit. Without a valid concern on overfitting and without a valid measure for complexity, there are no overfitting and complexity to be balanced. The penalty terms become unnecessary. The MLC which does not have a penalty term is a fairer criterion in this setting.


\end{document}